\documentclass[]{article}
\usepackage{proceed2e}

\usepackage{times}
\usepackage[ruled,vlined]{algorithm2e}
\usepackage{tabularx}

\usepackage{soul}
\usepackage{url}
\usepackage[hidelinks]{hyperref}
\usepackage[utf8]{inputenc}
\usepackage[small]{caption}
\usepackage{graphicx}
\usepackage{booktabs}
\usepackage{amsmath,amssymb, bm}
\usepackage{bbm}
\usepackage{enumitem}
\usepackage{bbm}
\usepackage[numbers]{natbib}
\usepackage{amsthm}
\usepackage{graphicx}
\usepackage{subcaption}
\usepackage{comment}
\usepackage[dvipsnames]{xcolor}

\newtheorem{lemm}{Lemma}

\title{Skewness Ranking Optimization for Personalized Recommendation}

\author{} 

\author{ {\bf Chuan-Ju Wang\thanks{~~These authors contributed equally to this work; author order was determined by seniority.}
}\\
Academia Sinica\\
Taipei, Taiwan\\
cjwang@citi.sinica.edu.tw
\And
{\bf Yu-Neng Chuang\footnotemark[1]
}  \\
National Chengchi University \\
Taipei, Taiwan\\
107753011@nccu.edu.tw
\And
{\bf Chih-Ming Chen\thanks{~~Social Networks and Human-Centered Computing, Taiwan International Graduate Program, Institute of Information Science, Academia Sinica, Taiwan.}
}   \\
National Chengchi University  \\
Taipei, Taiwan \\
104761501@nccu.edu.tw
\And
{\bf Ming-Feng Tsai}   \\
National Chengchi University \\
Taipei, Taiwan\\
mftsai@nccu.edu.tw
}

\begin{document}

\maketitle

\begin{abstract}

In this paper, we propose a novel optimization criterion that leverages
features of the skew normal distribution to better model the problem of
personalized recommendation. 
Specifically, the developed criterion borrows the concept and the flexibility
of the skew normal distribution, based on which three hyperparameters are
attached to the optimization criterion.    
Furthermore, from a theoretical point of view, we not only establish the
relation between the maximization of the proposed criterion and the shape
parameter in the skew normal distribution, but also provide the analogies and
asymptotic analysis of the proposed criterion to maximization of the area under the ROC curve. 
Experimental results conducted on a range of large-scale real-world datasets
show that our model significantly outperforms the state of the art and yields
consistently best performance on all tested datasets.
\end{abstract}

\section{INTRODUCTION}

Now ubiquitous, recommender systems are an
indispensable component of services and platforms such as music and video
streaming services and e-commerce websites.
Real-world recommender systems comprise a number of user-item interactions that
facilitate recommendations, including ratings, playing times, likes, sharing,
and tags.
In general, these interactions can be divided into explicit feedback
(e.g., in terms of ratings) and implicit feedback (e.g., monitoring clicks, view
times); in real-world scenarios, most feedback is not explicit but
implicit. 

Collaborative filtering (CF) is a commonly adopted approach that leverages
either explicit or implicit user-item interactions for item recommendation.
Many CF-based recommendation algorithms have been shown to yield reasonable
performance across various domains and have been used in many real-world
applications.
Among CF-based approaches, model-based CF has become a mainstream type of
recommendation algorithms, the core idea of which is to learn effective
low-dimensional dense representations of users and items from either explicit
or implicit feedback for recommendation.

In the model-based CF literature, latent factor models discover shared latent
factors (i.e., user/item representations) by decomposing a given user-item
interaction matrix, which has proven effective for explicit user feedback.
Matrix factorization is the most representative of this type of
approaches~\cite{mf,slim,fism}.
However, it is problematic to apply traditional matrix factorization to
implicit feedback as we can neither ignore unobserved user-item
interactions nor assume that these unobserved interactions are negative.
To address this, weighted regularized matrix factorization (WRMF)
proposed by~\cite{wrmf1,wrmf2} incorporates all the unobserved user-item
interactions as negative samples and uses a case weight to reduce the impact of
these uncertain samples.
Moreover, over the past decade, the focus of literature has shifted to optimizing
item ranks from implicit data as opposed to predicting explicit item
scores~\cite{bpr,cse,translation-based,cml,ngcf,cofactor,hop-rec,kgat}, namely
ranking-based recommendation approaches. 
Most of these approaches assume that unobserved items are of less interest to
users and are thus mainly designed to discriminate observed (positive) items
from unobserved (negative) items. 

Bayesian personalized ranking (BPR)~\cite{bpr} is a pioneering,
well-known example of ranking-based recommendation models. 
The authors propose a generic optimization criterion for personalized ranking
that maximizes the posterior probability of user preferences from pairs of
observed and unobserved items for each user.  
Later ranking-based studies such as WARP~\cite{warp} and K-OS~\cite{kos} adopt BPR's pair-wise
ranking concept, creating new variants by modifying the loss function
to better model the problem.
Moreover, for this Bayesian modeling approach to personalized ranking, these
models all leverage the assumption that the prior probability for the model
parameters is normally distributed.
Nevertheless, neither BPR itself nor later works closely investigate
the learned distribution of the estimator---a real-valued function of the
model parameters that captures the relationship between users and their observed
and unobserved items---which is however the component most related to model
performance.

Therefore, to better model the problem, we first study the learned
distributions of the estimator from different ranking-based methods, and we
observe that the realized distributions are in general unimodal and typically skewed.
As a result, we consider the skew normal distribution a good
candidate to better analyze and model the problem because of its generality. 
Particularly, there are two sides to our story.
First, we leverage features of the skew normal distribution to design a new
optimization criterion for personalized ranking.  
Second, with the assumption that the estimator follows the skew
normal distribution, we provide insights and theoretical results for the proposed optimization criterion.
Specifically, skewness ranking optimization (Skew-OPT), the optimization 
criterion we develop, is parameterized with three additional hyperparameters,
two of which are inspired by the location and scale parameters in the
skewness normal distribution and one of which is related to the shape of the
gradient function derived from the optimization objective, thereby providing additional degrees of freedom for ranking optimization.
With this design, we provide two theoretical results.
First, under the assumption that the estimator follows the skew normal
distribution with fixed location and scale parameters, maximization of the
proposed criterion simultaneously maximizes the shape parameter in the skew
normal distribution along with the skewness value of the distribution.
Second, we provide the analogies and asymptotic analysis of Skew-OPT to maximization of the area under the ROC curve. 

Extensive experiments were conducted on five representative and publicly
available recommendation datasets.
We compare our model with WRMF~\cite{wrmf1,wrmf2}, a matrix factorization based method for
implicit feedback; BPR~\cite{bpr} and WARP~\cite{warp}, two ranking-based methods;
HOP-Rec~\cite{hop-rec}, a state-of-the-art model that combines the
concept of latent factor and graph-based models; and
NGCF~\cite{ngcf}, a recent neural model for collaborative filtering. 
The evaluation shows that learning with the proposed Skew-OPT outperforms the
competing methods for all datasets, and the performance improvements are
significant by a large amount in terms of two commonly used top-$N$
recommendation evaluation metrics.
Particularly, for four out of the five datasets, our model achieves more than
10\% improvement compared to the best performing baseline models. 
In summary, the contributions of this work are:
\begin{itemize}[noitemsep]
\item We present Skew-OPT, a novel ranking optimization criterion that
significantly outperforms other state-of-the-art models for 
user-item recommendation.
\item To our best knowledge, this is the first attempt to leverage the features
and concept of the skew normal distribution to construct the optimization
criterion and analyze the model.
\item
The developed criterion is parameterized with three hyperparameters, providing
additional degrees of freedom for ranking optimization.
\item
We provide theoretical results on the relation between the maximization of
Skew-OPT and the shape parameter in the skew normal distribution as well as
the analogies and asymptotic analysis of the criterion to maximization
of the area under the ROC curve. 
\item
We report extensive experiments over five recommendation datasets to
demonstrate the robustness and effectiveness of the proposed
method.\footnote{For reproducibility, we will share the source code online at a
GitHub repo upon publication.}
\end{itemize}

\section{Personalized Recommender Systems}
The task of personalized recommendation is to provide a list of
ranked items to users based on their historical interactions with items.
Specifically, we investigate scenarios where the ranking is to be inferred
from the implicit user feedback. Below, we first formulate such a
personalized recommendation task with a representation learning approach in
Section~\ref{sec:formulation}. 
We then provide background knowledge regarding traditional Bayesian approaches
for personalized ranking, the skewness of a given distribution, and the skew
normal distribution in Section~\ref{sec:pre}.

\subsection{Formalization}\label{sec:formulation}
Let $U$ and $I$ be the sets of users and items, respectively.
Given user-item implicit feedback $S\subseteq U\times I$, our goal is to learn a
representation matrix $\Theta \in \mathbb{R}^{|U\cup I|\times|d|}$ for all
users and items such that for each user $u\in U$, we generate the top-$N$
recommended items by computing the dot products of $\theta_u$ and $\theta_i$
$\forall i\in I$, where $d$ denotes the dimension of the learned
representations, and $\theta_u$ and $\theta_i$ are the row vectors of $\Theta$
denoting the representations of user $u$ and item $i$, respectively. 
It is expected that the learned representation matrix $\Theta$ not only well matches
the observed user preferences but also predicts unobserved user
preferences.

\subsection{Preliminaries}\label{sec:pre}
\subsubsection{Bayesian Approaches for Personalized Ranking}
For personalized recommendation, conventional ranking-based
methods such as Bayesian personalized ranking (BPR)~\cite{bpr} propose 
modeling preference order by using item pairs as training data and optimizing
for the correct ranking of item pairs.
Such methods create a set of triple relations $D_S: U\times I\times I$ from
user feedback $S$ for model training by 
\begin{equation*}
D_S = \left\{ (u,i,j)\,|\,\forall u \in U, i \in I^{+}_{u}\wedge j \in I \setminus I^{+}_{u} \right\},
\end{equation*}
where $(u,i,j)\in D_S$ means that user $u$ is assumed to prefer item $i$ over
item $j$. 
For notational simplicity, we introduce notation $>_u$ to denote the pairwise
user preference for user~$u$; i.e., $i >_u j$ means that $u$ prefers item $i$
over $j$.
With the above construction, the generic optimization criterion for the
ranking-based methods is
\begin{eqnarray}
    \notag\ln P(\Theta\,|>_u) &\propto&  \ln P(>_u|\Theta)P(\Theta)\\
    \label{eq:bpr-opt} 
    &=&\ln\prod_{(u,i,j)\in D_S} P(i>_uj|\Theta)P(\Theta)\\
    \notag &=& \sum_{(u,i,j)\in D_S} \ln P(i>_uj|\Theta) + \ln P(\Theta)\\
    \notag &=& \sum_{(u,i,j)\in D_S} \ln g\left(\hat{x}_{uij}(\Theta)\right) - \lambda_\Theta \|\Theta\|^2,
\end{eqnarray}
where $\hat{x}_{uij}(\Theta)$ is an arbitrary real-valued function of the model
parameter matrix $\Theta$ capturing relationships between user $u$, item $i$,
and item $j$; $g(\cdot)$ is a function used to describe the likelihood
function $P(i>_u j|\Theta)$ for $(u,i,j)$; and $\lambda_\Theta$ is a hyperparameter
for regularization.
Note that the last equality also involves a distribution assumption on the prior
density $p(\Theta)$, which is a normal distribution with zero mean and
variance-covariance matrix $\Sigma_\Theta$ (i.e., $p(\Theta)\sim
N(0, \Sigma_\Theta)$).
For notational simplicity, below we occasionally omit argument
$\Theta$ from function $\hat{x}_{uij}$.
In BPR, $g(\cdot)$ is set to the logistic sigmoid function and the estimator
$\hat{x}_{uij}$ is decomposed to $\hat{x}_{ui}$ and $\hat{x}_{uj}$ as
$\hat{x}_{uij}=\hat{x}_{ui}-\hat{x}_{uj}$, where $\hat{x}_{ui}$ is defined as
the dot product of $\theta_u$ and $\theta_i$ (i.e.,
$\hat{x}_{ui}=\langle\theta_u,\theta_i\rangle$).
Similar to most prior art, in this paper, we follow these settings in our
model.

\subsubsection{Skewness} 

Skewness is a measure of symmetry---more precisely, the lack of symmetry---of
the probability distribution of a real-valued random variable about its mean,
the value of which can be positive, negative, or undefined.
Formally, the skewness value $\gamma$ of a random variable $X$ is the third
standardized moment, which is defined as
\begin{equation*} 
    \gamma = \mathbb{E}\left[
    \left(\frac{X-\mu}{s}\right)^3 
    \right],
\end{equation*}
where $\mu$ and $s$ denote the mean and the standard deviation of ${X}$,
respectively.
For a unimodal distribution (e.g., normal distribution), a negative skew commonly
indicates that the tail is on the left side of the distribution, and a positive
skew indicates that the tail is on the right.
In addition, a zero value signifies that the tails on both sides of the mean balance
out overall, which is always true for a symmetric distribution but can also be
true for an asymmetric distribution in which one tail is long and thin and the
other is short but fat.

\subsubsection{Skew Normal Distribution}\label{sec:skewnorml}
In probability theory and statistics, the skew normal distribution is a
continuous probability distribution that generalizes the normal distribution to
allow for non-zero skewness. 
Generally speaking, the probability density function (PDF) of a skew normal
distribution can be defined with parameters location $\xi\in
\mathbb{R}$, scale $\omega\in \mathbb{R}^+$, and shape $\alpha\in\mathbb{R}$:
\begin{equation} 
    f(x) = \frac{2}{\omega} \varphi\left(\frac{x-\xi}{\omega}\right) \Psi\left(\alpha\left(\frac{x-\xi}{\omega}\right)\right),
    \label{eq:skewnormaldensity}
\end{equation}
where $\varphi(\cdot)$ and $\Psi(\cdot)$ denote the PDF and the cumulative
distribution function (CDF) of the standard normal distribution, respectively. 
Moreover, the CDF of $X$ is 
\begin{equation}
F(x)=\Psi\left(\frac{x-\xi}{\omega}\right)-2T\left(\left(\frac{x-\xi}{\omega}\right),\alpha\right),
\label{eq:CDF}
\end{equation}
where $T(h,a)$ is Owen's $T$ function. 
Then the skewness value $\gamma$ of the skew normal distribution is a function
of $\alpha$ defined as
\begin{equation} 
    \gamma(\alpha) = \frac{4-\pi}{2} \frac{\left(\frac{\alpha}{\sqrt{1+\alpha^2}} \sqrt{\frac{2}{\pi}}\right)^3}{\left(1- \frac{2\alpha^2}{\pi(1+\alpha^2)}\right)^\frac{3}{2}}\,.
    \label{eq:skewnormal_skew}
\end{equation}

\begin{figure}
\begin{center}
\includegraphics[width=0.46\textwidth]{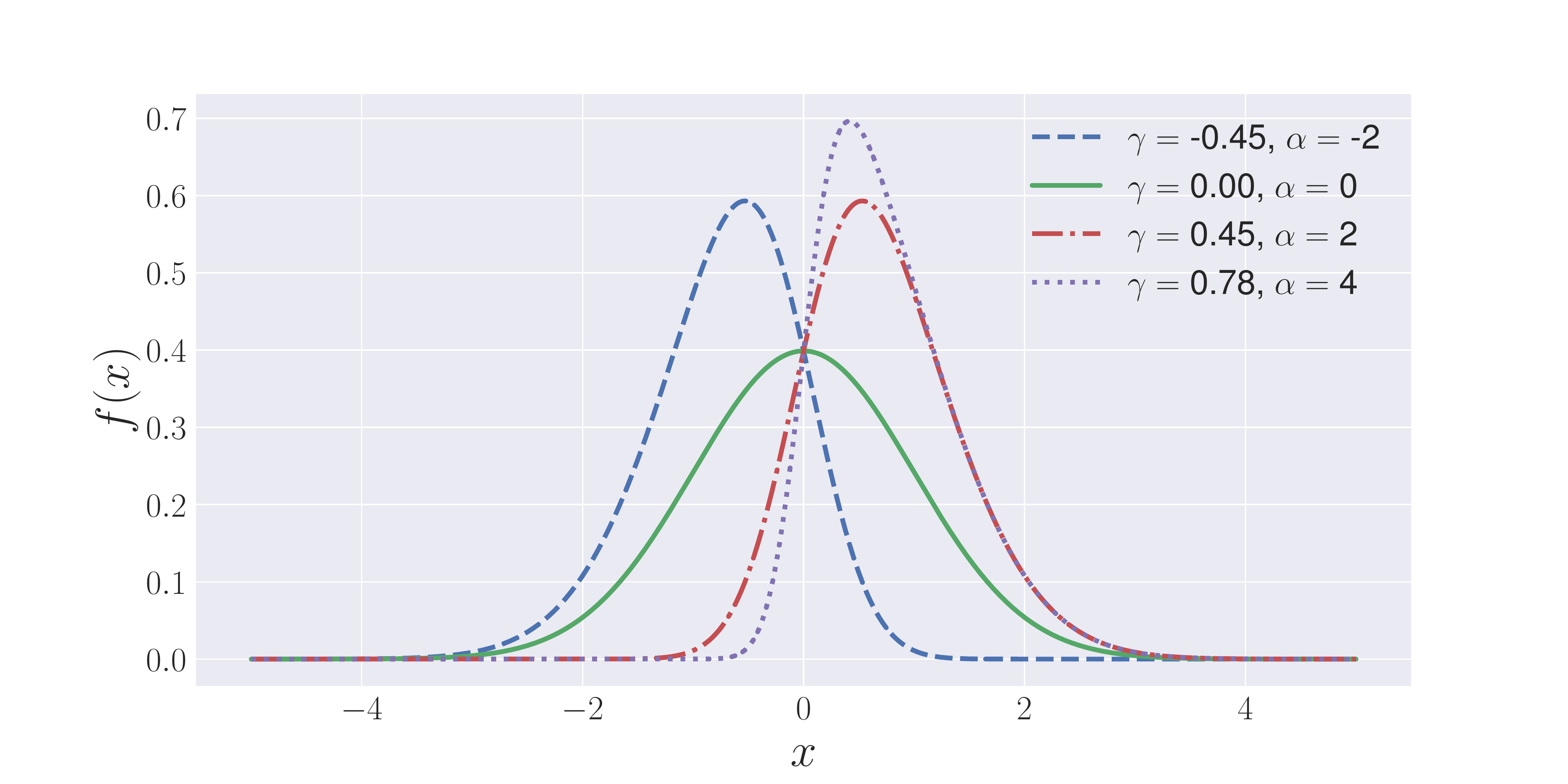}
\vspace{-0.2cm}
\caption{Skew normal distributions ($\xi=0,\omega=1$).}
\label{fig:skew_n}
\end{center}
\end{figure}

Figure~\ref{fig:skew_n} illustrates the PDFs of the skew normal distribution
with fixed location parameter $\xi=0$ and scale parameter $\omega=1$, but with
different shape parameters, i.e., $\alpha=-2,0,2,4$. 
From the figure, we observe that a larger $\alpha$ yields a larger skewness value
$\gamma$.
Moreover, with fixed $\xi$ and $\omega$, it is clear that enlarging $\alpha$
increases the probability $p(x>0)$; this argument will be later elaborated in
our method and linked to the metric AUC in Section~\ref{sec:AUC}.  

\section{Skewness Ranking Optimization (Skew-OPT)}\label{sec:method}

In this section, we first make some observations about the representations
learned from BPR, based on which we briefly explain the motivation for our work, in Section~\ref{sec:observation}.
Second, we provide a detailed derivation of the proposed optimization criterion in Section~\ref{sec:criterion}; finally Section~\ref{sec:AUC} gives the analogies between our criterion and the AUC evaluation
metric.



\subsection{Observation and Motivation}\label{sec:observation}

Most prior art for personalized ranking, such BPR~\cite{bpr} and
WARP~\cite{warp}, seeks to learn effective user and item representations for
item recommendation by maximizing the posterior probability of user
preferences from pairs of observed and unobserved items for each user.
Among these methods, BPR, the most representative work, introduces a general
prior density $p(\Theta)$ that follows a normal distribution with zero mean and
variance-covariance matrix $\lambda_\Theta I$ to complete the Bayesian modeling
approach of the personalized ranking task.
Nevertheless, neither BPR itself nor its succeeding works discuss the
distribution of the estimator (i.e., $\hat{x}_{uij}(\Theta)$), which is however
the component most related to model performance.
Figure~\ref{fig:bprorg} plots the distribution constructed by the learned
estimates for $\hat{x}_{uij}(\Theta)$ with the use of BPR training on each of
the three listed datasets.
From the figure, we observe that the three distributions are unimodal in
general and typically skewed---the distributions for Epinions-Extend and Last.fm-360K 
are
right-skewed with positive sample skewness values ($\hat{\gamma}=1.09$ and $\hat{\gamma}=0.373$ in panel
(a) and (b), respectively), and that for Amazon-Book is almost symmetric with a close-to-zero positive skewness value ($\hat{\gamma}=0.08$ in panel (c)).

\begin{figure}[!h]
    \centering
    \begin{subfigure}[b]{0.145\textwidth}
        \includegraphics[width=\textwidth]{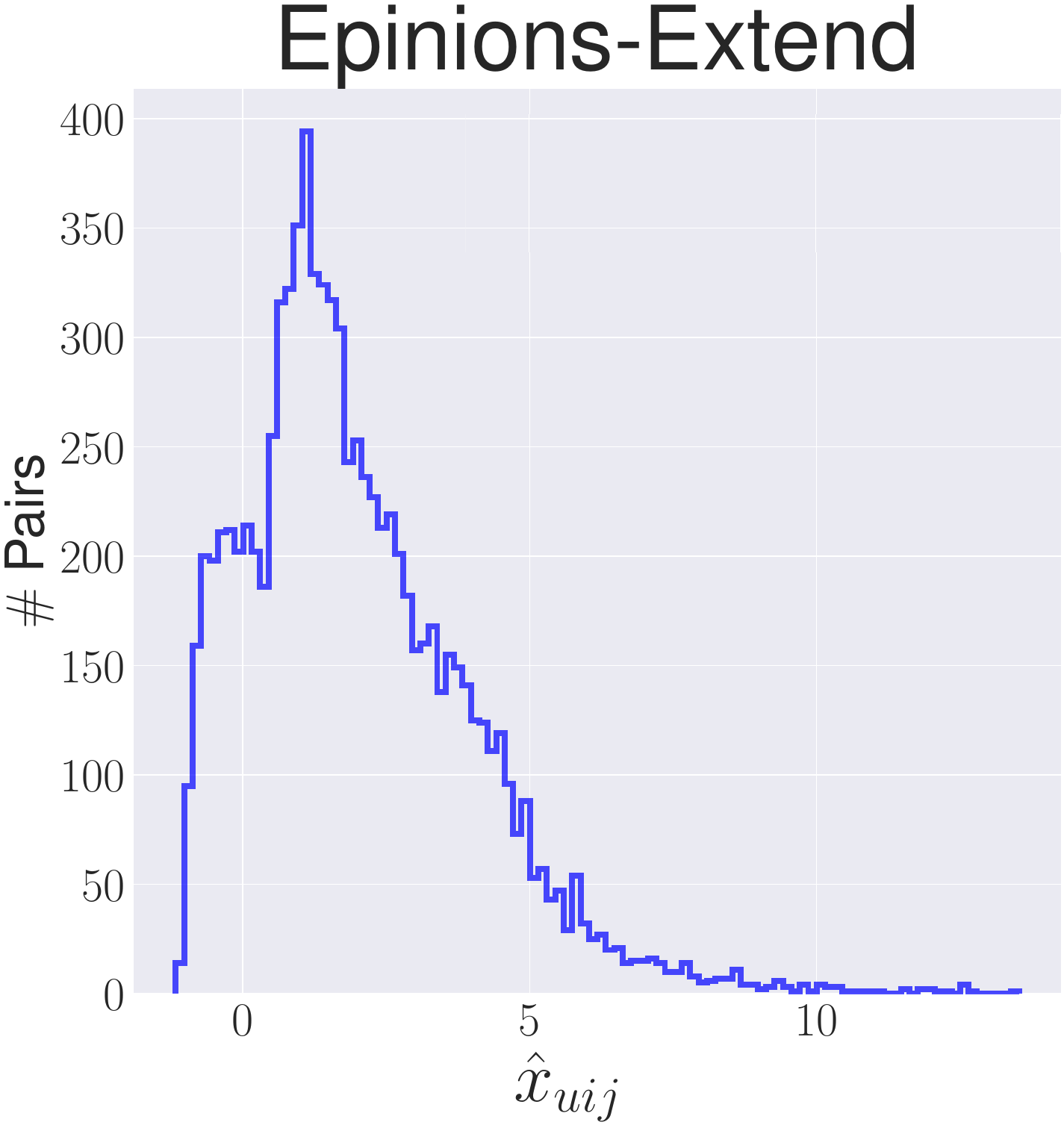}
        \caption{$\hat{\gamma} = 1.09$}
        \label{fig:rbpr}
    \end{subfigure}
    ~
    \begin{subfigure}[b]{0.145\textwidth}
        \includegraphics[width=\textwidth]{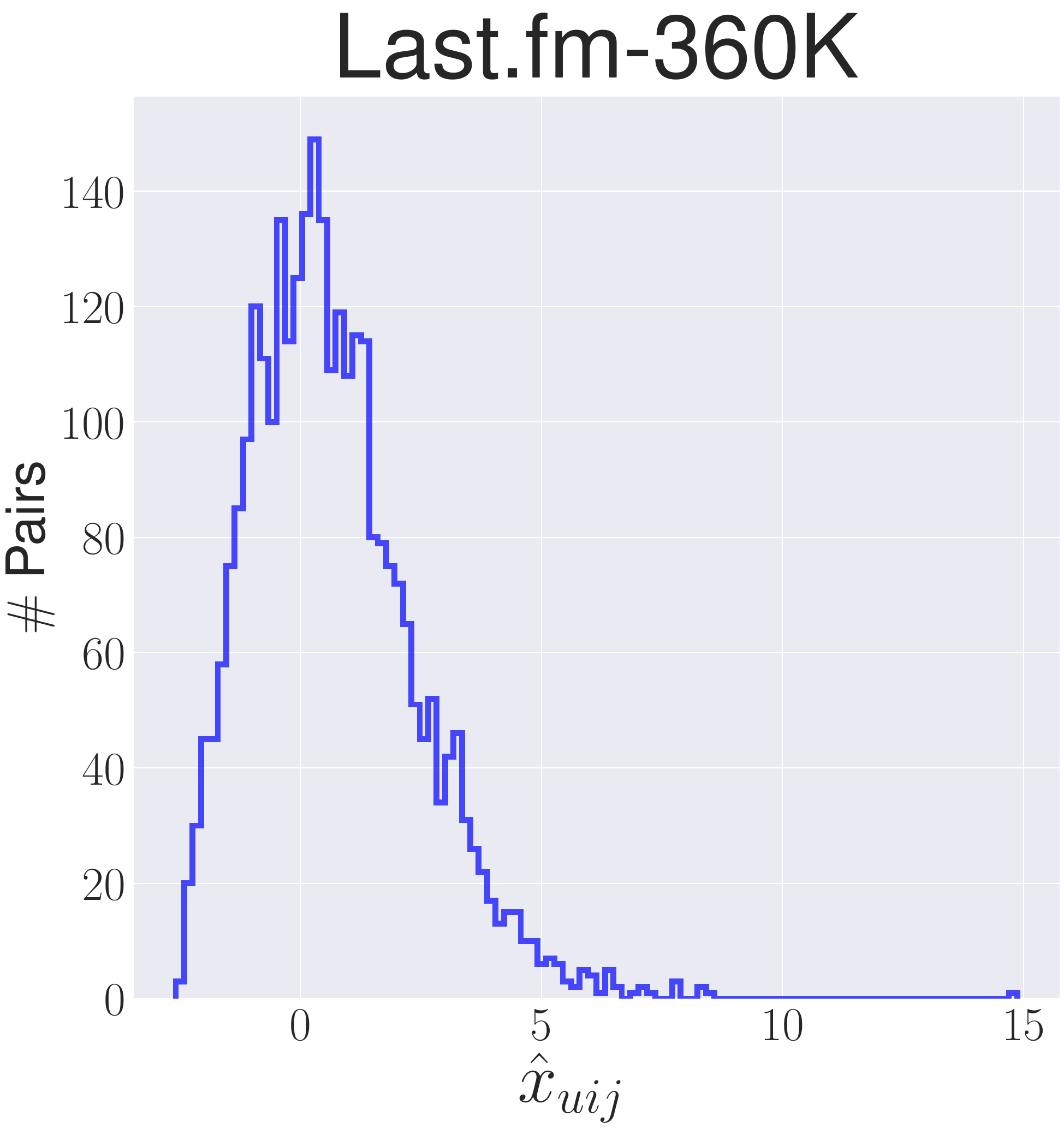}
        \caption{$\hat{\gamma}=0.373 $}
        \label{fig:lbpr}
    \end{subfigure}
    ~
    \begin{subfigure}[b]{0.145\textwidth}
        \includegraphics[width=\textwidth]{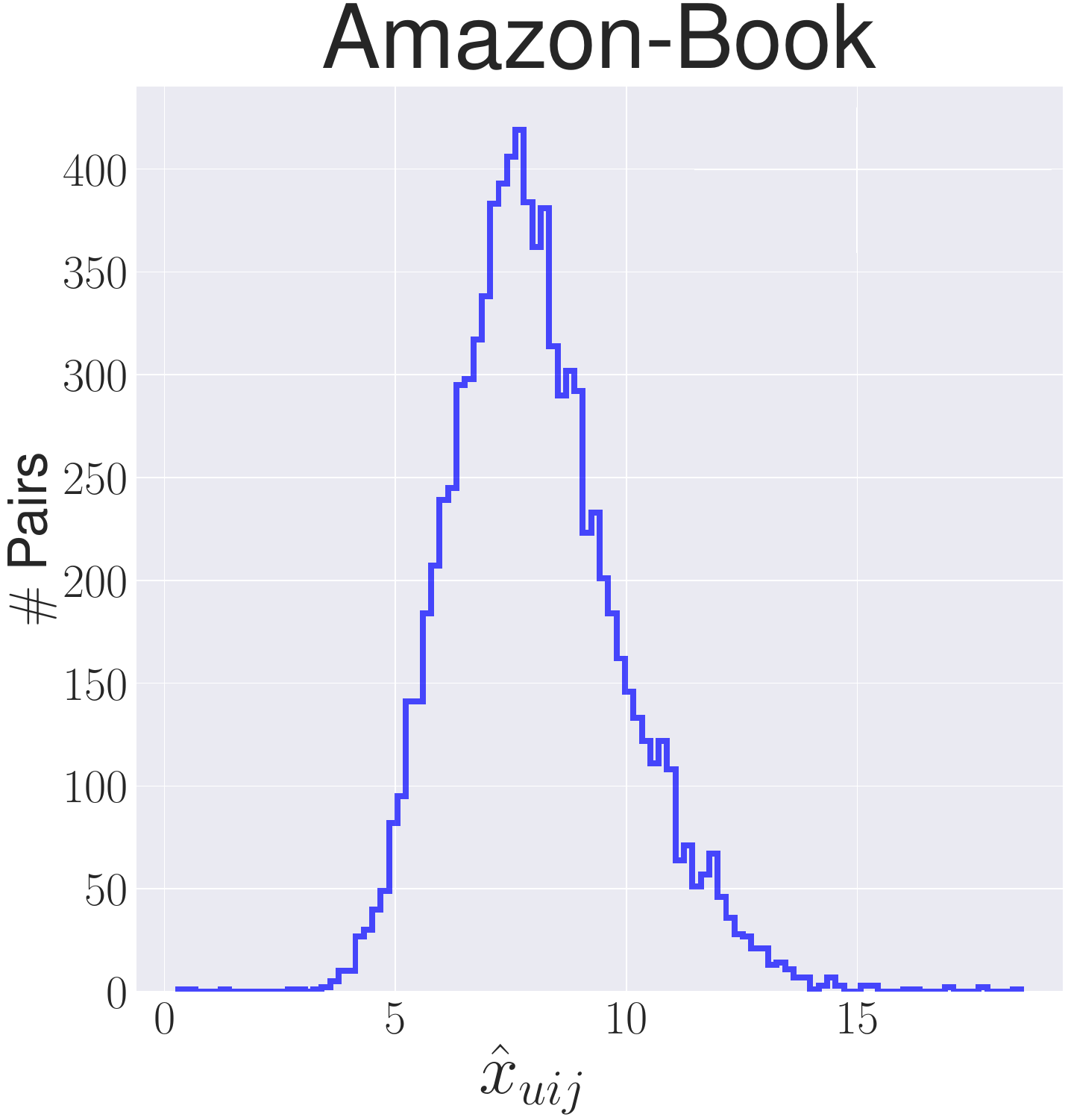}
        \caption{$\hat{\gamma} = 0.08$}
        \label{fig:sbpr}
    \end{subfigure}
    ~
    \caption{Distributions of $\hat{x}_{uij}$ learned from BPR.}\label{fig:bprorg}
\end{figure}

Inspired by the above observations (e.g., unimodal and skewed distributions), in this paper, we propose a simple yet novel optimization criterion
that leverages features of the skew normal distribution to better model
the problem.
First, the location parameter $\xi$ in the skewness normal distribution
provides an additional degree of freedom to allow us push the distribution of the
estimator to the right; also, the scale parameter $\omega$ is used to 
reduce model over-fitting for large $\xi$.
In addition, from Figure~\ref{fig:skew_n}, with a fixed $\xi$ and $\omega$,
enlarging the shape parameter $\alpha$ increases the probability $p(x>0)$.
Here, for personalized ranking, the random variable $X$ can be used
to describe the estimator $\hat{x}_{uij}$; thus, in this case, a larger
$\alpha$ entails a larger probability $p(\hat{x}_{uij}>0)$, which should 
benefit recommendation performance. 
Details for the proposed optimization criterion and its link to the AUC are provided in
Sections~\ref{sec:criterion} and~\ref{sec:AUC}, respectively.

\subsection{Criterion and Optimization}\label{sec:criterion}
\label{SPR}
Motivated by the above observations as well as the properties of the skew normal
distribution, in this paper we propose an unconventional optimization criterion
termed skewness ranking optimization (Skew-OPT) for personalized
recommendation. 
To this end, we recast the likelihood function referring to the individual
probability that a user really prefers item $i$ to item $j$ in
Eq.~(\ref{eq:bpr-opt}) as  
\begin{equation}
     p(i >_{u} j\,|\Theta,(\xi,\omega,\eta)) = \sigma\left(\left(\frac{{\hat{x}}_{uij}(\Theta)-\xi}{\omega}\right)^\eta\right),
\label{eq:newlikelihood}
\end{equation}
where $(\xi,\omega,\eta)$ denote three hyperparameters in the proposed
Skew-OPT, $\eta\in\mathbb{O}$, and $\sigma(\cdot)$ denotes the sigmoid
function.
Above, the inclusion of $\xi$ and $\omega$ is motivated by the location and scale parameters in
the skew normal distribution, respectively (see Section~\ref{sec:skewnorml}),
and $\mathbb{O}$ denotes the set of positive odd integers.
Note that forcing $\eta$ to be a positive odd integer ensures the rationality
of the likelihood function, as under this setting it is an increasing function
with argument ${\hat{x}}_{uij}$ (i.e., the distance between an observed item
and a non-observed one).
As mentioned previously, the location parameter $\xi$ here provides an
additional degree of freedom to allow us push the distribution of the estimator
to the right, and the scale parameter $\omega$ can be used to reduce
overfitting for large $\xi$.
It is also worth mentioning that the likelihood of BPR is a special case of
Eq.~(\ref{eq:newlikelihood})
with $\xi=0,\omega=1,\eta=1$.

With the above likelihood function in Eq.~(\ref{eq:newlikelihood}), the
optimization criterion becomes maximizing
\begin{align}
    \notag &\text{Skew-OPT}\\
    \notag :=& \ln\prod_{(u,i,j)\in D_S} p\left(i>_uj|\Theta,(\xi,\omega,\eta)\right)\,p(\Theta)\\
    \notag =& \sum_{(u,i,j)\in D_S} \ln p\left(i>_uj|\Theta,(\xi,\omega,\eta)\right) + \ln p(\Theta)\\
    \label{eq:Skew-OPT}
    =&\sum_{(u,i,j)\in D_S} \ln \sigma\left(\left(\frac{{\hat{x}}_{uij}(\Theta)-\xi}{\omega}\right)^\eta\right)
       - \lambda_{\Theta}\|\Theta\|^2.
\end{align}

Now, we discuss the relationship between Skew-OPT optimization and the
shape parameter $\alpha$ and the corresponding skewness value.

\begin{figure}
\begin{center}
\includegraphics[width=0.46\textwidth]{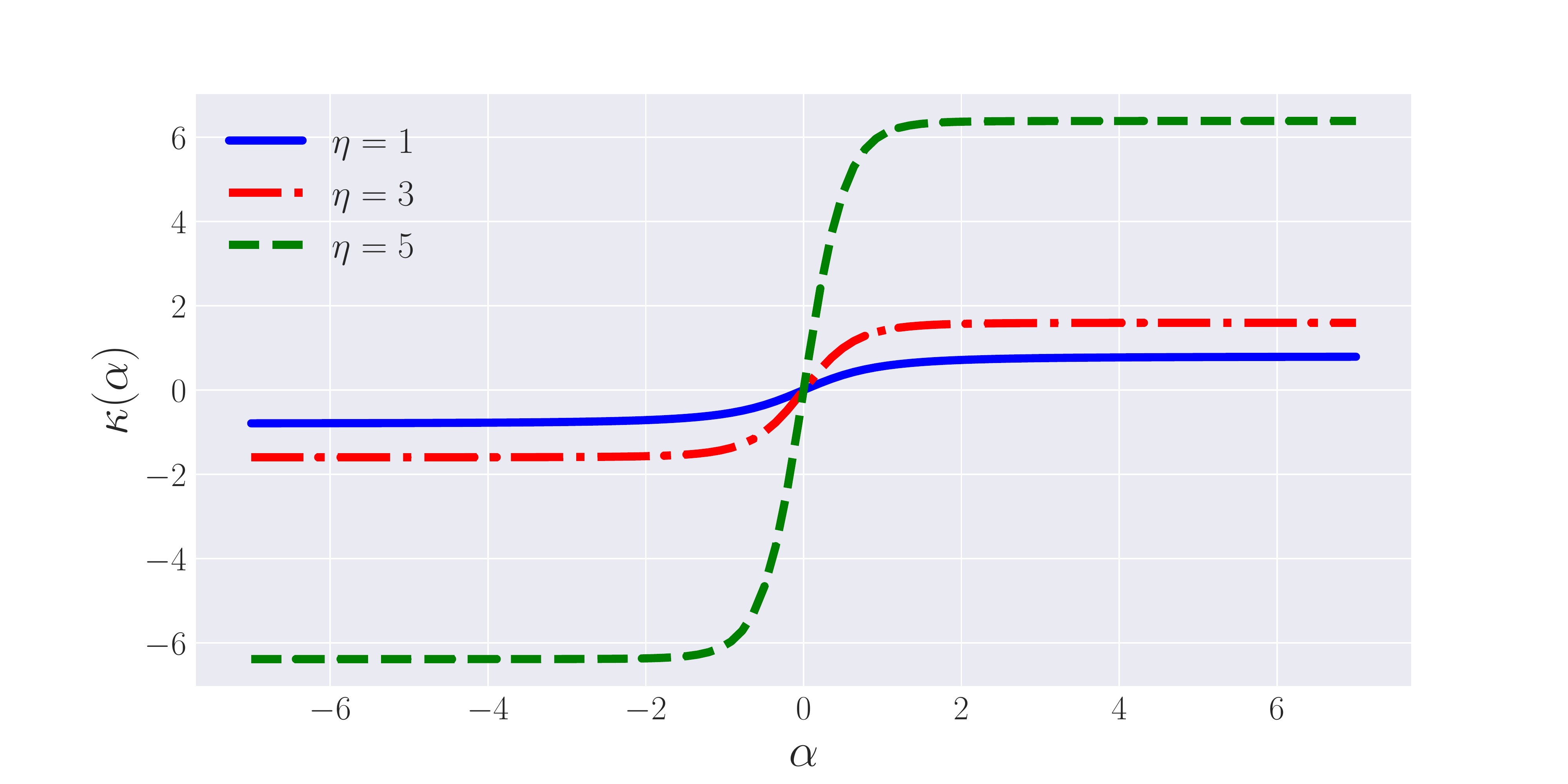}
\vspace{-0.2cm}
\caption{Increasing function $\kappa(\alpha)$ ($\xi=0,\omega=1$).}
\label{fig:kappa}
\end{center}
\end{figure}

\begin{lemm}
\label{lm1}
Given the case that $\hat{x}_{uij}$ follows a skew normal distribution with
fixed location parameter $\xi$ and scale parameter $\omega$, maximizing the first term of
Eq.~(\ref{eq:Skew-OPT}) for a certain~$\eta$ simultaneously maximizes the shape parameter~$\alpha$
and the skewness value of the estimator, $\hat{x}_{uij}(\Theta)$.
\end{lemm}
\begin{proof}
In Eq.~(\ref{eq:Skew-OPT}), the first term can be written as 
\begin{equation*}
     \sum_{(u,i,j)\in D_S} -\ln \left(1+e^{-\left(\frac{{\hat{x}}_{uij}(\Theta)-\xi}{\omega}\right)^\eta}\right).
     \label{eq:first-term}
\end{equation*}
Omitting the $1$ in the above equation makes it clear that
maximizing the above summation is equivalent to maximizing
\begin{align}
      \notag &\sum_{(u,i,j)\in D_S} \left(\frac{{\hat{x}}_{uij}(\Theta)-\xi}{\omega}\right)^\eta\\
      &\propto\mathbb{E}_{(u,i,j)\sim D_S} \left[\left(\frac{{\hat{x}}_{uij}(\Theta)-\xi}{\omega}\right)^\eta\right].
      \label{eq:alpha_function}
\end{align}
With fixed $\xi$, $\omega$, and $\eta$, when $\hat{x}_{uij}$ follows a skew
normal distribution, Eq.~(\ref{eq:alpha_function}) can be represented as a
function of the shape parameter $\alpha$ as
\begin{align}
    \label{eq:alpha_expectation}
    \kappa(\alpha)&=\mathbb{E}
    \left[\left(\frac{{\hat{x}}_{uij}(\Theta)-\xi}{\omega}\right)^\eta\right],
\end{align}

Now, we prove that both $\kappa(\alpha)$ in Eq.~(\ref{eq:alpha_expectation})
and $\gamma(\alpha)$ in Eq.~(\ref{eq:skewnormal_skew}) are increasing functions
by showing that $\partial\kappa(\alpha)/\partial\alpha>0$ and
$\partial\gamma(\alpha)/\partial\alpha>0$.
For the former, we have
\begin{align}
&\notag\partial\kappa(\alpha)/\partial\alpha\\
\notag =&\partial\left(\int_{-\infty}^\infty\left(\frac{x-\xi}{\omega}\right)^\eta f(x) dx \right)/\partial\alpha\\
\notag =&\int_{-\infty}^\infty\left(\frac{x-\xi}{\omega}\right)^\eta \partial f(x)/ \partial\alpha\, dx\\
\label{eq:kappa}
=&\int_{-\infty}^\infty\left(\frac{x-\xi}{\omega}\right)^{\eta+1}\left(\frac{2}{\omega}\right)\phi\left(\frac{x-\xi}{\omega}\right) \left(\frac{e^{-\frac{\alpha^2\left(x-\xi\right)^2}{2\omega^2}}}{\sqrt{2\pi}}\right) dx.
\end{align}
where the density function $f(x)$ is defined in Eq.~(\ref{eq:skewnormaldensity}).

Above, the first component in Eq.~(\ref{eq:kappa}) is greater than or equal to
zero as $\eta+1$ is an even integer; the remaining three components are all positive
as $\omega>0$, $\phi(\cdot)$ is a PDF, and the numerator of the last component
is an exponential function.
Moreover, since Eq.~(\ref{eq:kappa}) involves integration over all $x$, it is
clear that we have $\partial\kappa(\alpha)/\partial\alpha>0$, and thus
$\kappa(\alpha)$ is an increasing function (see Figure~\ref{fig:kappa} for
example).
Similarly, it is easy to prove that
$\partial\gamma(\alpha)/\partial\alpha>0$, an illustration for which is shown in
Figure~\ref{fig:gamma}.
As a result, a larger expected value in Eq.~(\ref{eq:alpha_expectation})
corresponds to a larger $\alpha$; also, the value of skewness $\gamma$
increases as $\alpha$ increases, suggesting that maximizing the first term in
Eq.~(\ref{eq:Skew-OPT}) happens to simultaneously maximize the shape parameter
$\alpha$ along with the skewness of the estimator, $\hat{x}_{uij}(\Theta)$.
\end{proof}

\begin{figure}
\begin{center}
\includegraphics[width=0.46\textwidth]{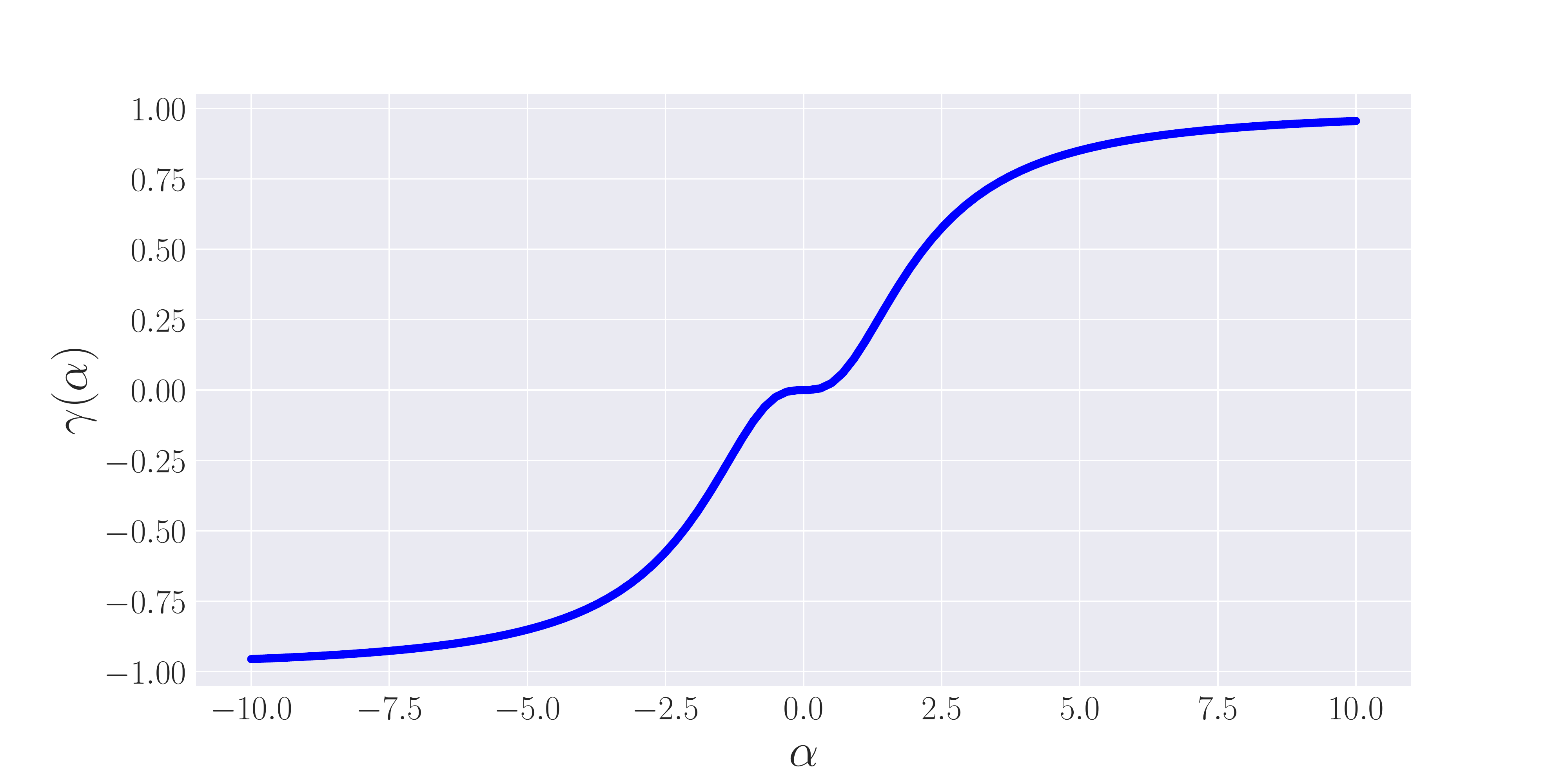}
\vspace{-0.2cm}
\caption{Increasing function $\gamma(\alpha)$.}
\label{fig:gamma}
\end{center}
\end{figure}

In the optimization stage, the objective function is maximized by utilizing the
asynchronous stochastic gradient ascent---the opposite of
asynchronous stochastic gradient descent (ASGD)~\cite{asgd}---for updating the
parameters $\Theta$ in parallel. For each triple $(u,i,j) \in D_{S}$, an update
with learning rate $\beta$ is performed as follows (see
Algorithm~\ref{alg:Skew-OPT}): 
\begin{equation*}
   \Theta \xleftarrow{} \Theta + \beta 
   \bigg( 
   \frac{\partial \text{Skew-OPT}}{\partial \Theta}
   \bigg),
\end{equation*}
where the gradient of Skew-OPT with respect to the model parameters is
\begin{align*}
\notag &\frac{\partial \text{Skew-OPT}}{\partial \Theta}\\
\notag =&\sum_{(u,i,j)\in D_S} \frac{\partial}{\partial \Theta} \ln \sigma\left(\left(\frac{{\hat{x}}_{uij}(\Theta)-\xi}{\omega}\right)^\eta\right)
       - \lambda_{\Theta}\|\Theta\|^2\\
       \propto&\sum_{(u,i,j)\in D_S} \frac{e^{-\left(\frac{{\hat{x}}_{uij}(\Theta)-\xi}{\omega}\right)^\eta}}{1+e^{-\left(\frac{{\hat{x}}_{uij}(\Theta)-\xi}{\omega}\right)^\eta}}\frac{\partial}{\partial \Theta}\left(\frac{\hat{x}_{uij}-\xi}{\omega}\right)^\eta-\lambda_\Theta\Theta.
\end{align*}

\begin{algorithm}[h!]
\SetAlgoLined
Input $D_S$\;
\Begin{
    Initialize $\Theta$\;
    \Repeat{\rm convergence}{
    Sample a triple $(u,i,j)$ from $D_S$\;
    $\Theta \xleftarrow{} \Theta + \beta \bigg( 
    \frac{\partial \text{Skew-OPT}\big(\hat{x}_{uij}(\Theta)\big)}{\partial\Theta} 
    \bigg)$;
    }
    \Return ${\Theta}$\;
}
 \caption{Model learning with Skew-OPT}
 \label{alg:Skew-OPT}
\end{algorithm}


\subsection{Analogies to AUC optimization}\label{sec:AUC}
With our optimization formulation in
Section~\ref{sec:criterion}, we here analyze the relationship between
Skew-OPT and AUC. 
The AUC per user is commonly defined as
\begin{equation*}
    {{\rm AUC}(u)} := \frac{1}{\left| I_u^+ \right|\left| I\setminus I_u^+ \right| }
    \sum_{i \in I_u^+}\sum_{j\in I\setminus I_u^+} \delta (\hat{x}_{uij} > 0).
\end{equation*}
Above, $\delta (x_{uij})$ is the indicator function defined as
\begin{equation*}
    \delta (\hat{x}_{uij}) = 
     \begin{cases}
        \; 1, \; \text{if} \;\;\hat{x}_{uij} > 0 \\
        \; 0, \; \text{otherwise}.\\ 
     \end{cases}
\end{equation*}
The average AUC of all users is
\begin{align}
    \notag {\rm AUC} :=& \frac{1}{\left| U \right|} \sum_{u \in U} {{\rm AUC}(u)}\\
    =& 
    \sum_{(u,i,j)\in D_S} w_u \delta (\hat{x}_{uij} > 0), \label{eq:AUC}
\end{align}
where
\begin{equation*}
w_u=\frac{1}{\left|U\right|\left|I_u^+ \right|\left|I\setminus I_u^+ \right|}.
\end{equation*}

The analogy between Eq.~(\ref{eq:AUC}) and the objective function of BPR is
clear as their main difference is the normalizing constant. 
Note that BPR is a special case with $\xi=0,\omega=1,\eta=1$ in the proposed
Skew-OPT.
With Skew-OPT, the analogy becomes a bit involved and is explained as follows.
In the proposed Skew-OPT with fixed hyperparameters $\xi,\omega,\eta$,
Lemma~\ref{lm1} states that maximizing the first term of
Eq.~(\ref{eq:Skew-OPT}) simultaneously maximizes the shape parameter~$\alpha$
under the assumption of the skew normal distribution for the estimator.
Moreover, as mentioned in Sections~\ref{sec:skewnorml}
and~\ref{sec:observation}, it is clear that increasing $\alpha$ enlarges the
probability $p(\hat{x}_{uij}>0)$, which is equal to the area under the PDF
curve for $\hat{x}_{uij}>0$.
This characteristic hence clearly shows the analogy between
Eq.~(\ref{eq:AUC}) and Skew-OPT. 

Whereas the AUC above refers to the macro average of the AUC
values for all users, we here consider the micro average version defined as
\begin{equation}
   {\rm AUC}^{\rm micro}:=\frac{1}{|D_S|}\sum_{(u,i,j)\in D_S} \delta(\hat{x}_{uij}>0).
   \label{eq:AUC-mirco}
\end{equation}
Under the assumption that $\hat{x}_{uij}$ follows the skew normal distribution
with fixed location parameter $\xi$ and scale parameter $\omega$, Eq.~(\ref{eq:AUC-mirco}) can be
rewritten as 
\begin{align*}
    {\rm AUC}^{\rm micro}:= & \mathbb{E}\,[\delta(\hat{x}_{uij}>0)\,]=p(\hat{x}_{uij}>0)\\
    = & 1-F(0)\\
    = & 1-\Psi\left(\frac{0-\xi}{\omega}\right)+2T\left(\left(\frac{0-\xi}{\omega}\right),\alpha\right),
\end{align*}
where $F(x)$ is the CDF of the skew normal distribution defined in Eq.~(\ref{eq:CDF}).
Additionally, when $\alpha\rightarrow\infty$, AUC$^{\rm micro}$ achieves its
maximum value, one, with $\xi\geq 0$, because
\begin{align}
    \forall \xi\geq 0,\,\, & \notag \lim_{\alpha\rightarrow\infty}2T\left(\left(\frac{0-\xi}{\omega}\right),\alpha\right)\\
    \label{eq:xigeq0}
    & = \frac{1}{2} \left(1+{\rm erf}\left(\frac{0-\xi}{\omega}\right)/\sqrt{2}\right)\\
    \notag & =  \Psi\left(\frac{0-\xi}{\omega}\right).
\end{align}
Note that for $\xi<0$, the limit value in Eq.~(\ref{eq:xigeq0}) becomes
$\frac{1}{2} \left(1-{\rm
erf}\left(\frac{0-\xi}{\omega}\right)/\sqrt{2}\right)$, but in this paper we do
not consider this case as we seek to maximize the estimator by shifting the 
distribution to the right on the horizontal axis.

\section{Experiments}
\subsection{Dataset}
To examine the performance of the proposed method, we conducted experiments on
five real-world datasets with different sizes, densities, and domains, the statistics ow which are shown
in Table~\ref{tab:addlabel}.
For each of the datasets, we converted the user-item interactions into implicit
feedback. For the 5-star rating datasets, we treated ratings higher than or
equal to 3.5 as positive feedback and the rest as negative feedback; as for the
count-based datasets, we took counts higher than 3 as positive
feedback and the remaining ones as negative feedback; for the CiteUlike dataset, since it is
already composed of binary user preferences, no transformation was needed.

\begin{table}
\centering
\setlength{\tabcolsep}{1mm}{
\small\small
\begin{tabular}{l rr rr c}
\toprule
& Users & Items & Edges & Edge type \\
\midrule 
CiteULike & 5,551 & 16,980 & 210,504 & like/dislike  \\
Amazon-Book & 70,679 & 24,916 & 846,522 & 5-star \\
Last.fm-360K & 23,566 & 48,123 &  303,4763 & play count \\
MovieLens-Latest & 259,137 & 40,110 &24,404,096 & 5-star \\
Epinions-Extend & 701,498 & 110,235 &12,581,748 & 5-star \\
\bottomrule
\end{tabular}}
\vspace{-0.2cm}
\caption{Dataset statistics}
\label{tab:addlabel}%
\end{table}%

\begin{table*}[!ht]
\centering
\makebox[\textwidth][c]{
\resizebox{1\textwidth}{!}{
\begin{tabular}{l rrr rrrr rrr rrr rrr}
\toprule
& \multicolumn{2}{c}{CiteUlike} & \multicolumn{2}{c}{Amazon-Book} & \multicolumn{2}{c}{Last.fm-360K} &
\multicolumn{2}{c}{MovieLens-Latest} & 
\multicolumn{2}{c}{Epinions-Extend} & \\
\cmidrule(lr){2-3}\cmidrule(lr){4-5}\cmidrule(lr){6-7}\cmidrule(lr){8-9}\cmidrule(lr){10-11} 
& Recall@10 & mAP@10 & Recall@10 & mAP@10 & Recall@10 & mAP@10 & Recall@10 & mAP@10 & Recall@10 & mAP@10 \\
\midrule
WRMF \cite{wrmf2, wrmf1} & 0.2159 & 0.1236 & 0.0950 & 0.0374 & 0.1308 & 0.0576 & 0.2122 & 0.1061 & 0.1025 & 0.0415 \\
BPR \cite{bpr} & 0.2217 & 0.1332 & 0.0972 & 0.0390 & 0.1394 & 0.0690 & 0.1952 & 0.1097 & 0.1137 & 0.0584 \\
WARP \cite{warp} & 0.1859 & 0.1033 & 0.0869 & 0.0356 & $\dagger$ 0.1763 & $\dagger$ 0.0937 & $\dagger$ 0.2748 & $\dagger$ 0.1634 &  0.1479 & 0.0711 \\
Hop-Rec \cite{hop-rec} & 0.2232 & 0.1319 & $\dagger$ 0.1072 & $\dagger$ 0.0426 & 0.1701 & 0.0870 & 0.2557 & 0.1419 & $\dagger$ 0.1617 & $\dagger$ 0.0813 \\
NGCF \cite{ngcf} & $\dagger$ 0.2321 & $\dagger$ 0.1367 & 0.0818 & 0.0335 & - & - & - & - & - & - \\
\midrule
Skew-OPT ($\eta=1$) & *{0.2413} & *{0.1541} & 0.1069 & *0.0467 & *0.1976 & *0.1051 & 0.2809 & 0.1636 & *0.1743 & *0.0914 \\
Improv. ($\%$) & +3.96\% & +12.72\%  & -0.27\%  & +9.62\% & +12.08\% & +12.17\%  & +2.21\% & +0.12\% & +7.79\% & +12.42\% \\
\midrule
Skew-OPT ($\eta=3$) & *{0.2481} & *{0.1591} & *\textbf{0.1173} & *0.0504 & *\textbf{0.2032} & *\textbf{0.1103} & *0.2852 & *0.1686 & *\textbf{0.1768} & *\textbf{0.0941} \\
Improv. ($\%$) & +6.89\% & +16.38\%  & +9.42\%  & +18.07\% & +15.25\% & +17.71\%  & +3.78\% & +3.18\% & +9.33\% & +15.74\% \\
\midrule
Skew-OPT ($\eta=5$) & *\textbf{0.2553} & *\textbf{0.1626} & *0.1163 & *\textbf{0.0522} & *0.2012 & *0.1083 & *\textbf{0.2879} & *\textbf{0.1699} & *0.1758 & *0.0915 \\
Improv. ($\%$) & +9.91\% & +18.94\%  & +8.48\%  & +22.53\% & +14.12\% & +15.58\%  & +4.76\% & +3.97\% & +8.71\% & +12.54\% \\
\bottomrule
\end{tabular}%
}}
\vspace{-0.2cm}
\caption{Recommendation performance. The † symbol indicates the best performing
score among all the compared models; ‘*’ and ‘Improv. (\%)’ denote statistical
significance at $p\text{-value} < 0.01$ with a paired t-test and the percentage
improvement of the proposed model, respectively, with respect to the best
performing value in the baselines.}
\label{tab:rec}%
\end{table*}%

\subsection{Baseline Algorithms}
In the following experiments, we compared our proposed model with the following
five representative and widely used recommendation algorithms.
\begin{itemize}[noitemsep]
    \item
	 {\bf WRMF}~\cite{wrmf2, wrmf1} (weighted regularized matrix factorization)
	 a relational weighted version of matrix factorization optimized by
	 utilizing least-square learning with an addition regularization term.
    \item
	 {\bf BPR}~\cite{bpr} (Bayesian personalized ranking) adopts pairwise
	 ranking loss for personalized recommendation and exploits direct user-item
	 interactions to separate negative items from positive items.
    \item 
	 {\bf WAPR}~\cite{warp} (weighted approximate-rank pairwise) an improved
	 ranking-based embedding model based on BPR, which weighs pairwise
	 violations depending on their position in the ranked list.
    \item
	 {\bf Hop-Rec}~\cite{hop-rec} (high-order proximity recommendation) 
	 a state-of-the-art hybrid model that integrates the concepts of 
	 graph-based and factorization-based models, where high-order neighbors
	 in a user-item interaction graph are exploited to enrich the information.
    \item
	 {\bf NGCF}~\cite{ngcf} (neural graph collaborative filtering) the
	 state-of-the-art neural-based CF model that recursively propagates the
	 embeddings on the user-item interaction graph, where high-order
	 connectivity is also encoded into user and item embeddings.
\end{itemize}

\subsection{Evaluation and Settings}
In the experiments, we focus on top-$N$ item recommendation.
To evaluate the model capability for this task, we utilized the following two
commonly used performance evaluation metrics: 1) recall and 2) mean average
precision (mAP).
For all datasets, we randomly divided the interaction data into 80$\%$ and
20$\%$ as the training set and the testing set, respectively. Also, the reported results
are the averaged results over five repetitions in this manner.
In addition, the dimensions of embedding vectors were all fixed to 128, and all
the hyperparameters of compared models were determined via a grid search over
different settings, from which the combination that leads to the best performance was
chosen. The ranges of hyperparameters we searched for the compared methods are
listed as follows.



\subsection{Experimental Results}

In the following sections, we demonstrate the recommendation performance and
several characteristics of the proposed Skew-OPT.
First, we conduct the experiments on the task of top-$N$ recommendation and
compare the proposed method with the five baselines. 
We then provide a sensitivity analysis for the three key hyperparameters in
our model.
Finally, we study the learned distributions of the estimator for the five
datasets and compare them with the skew normal distribution.

\begin{figure*}[htb]
\begin{center}
\includegraphics[width=\textwidth]{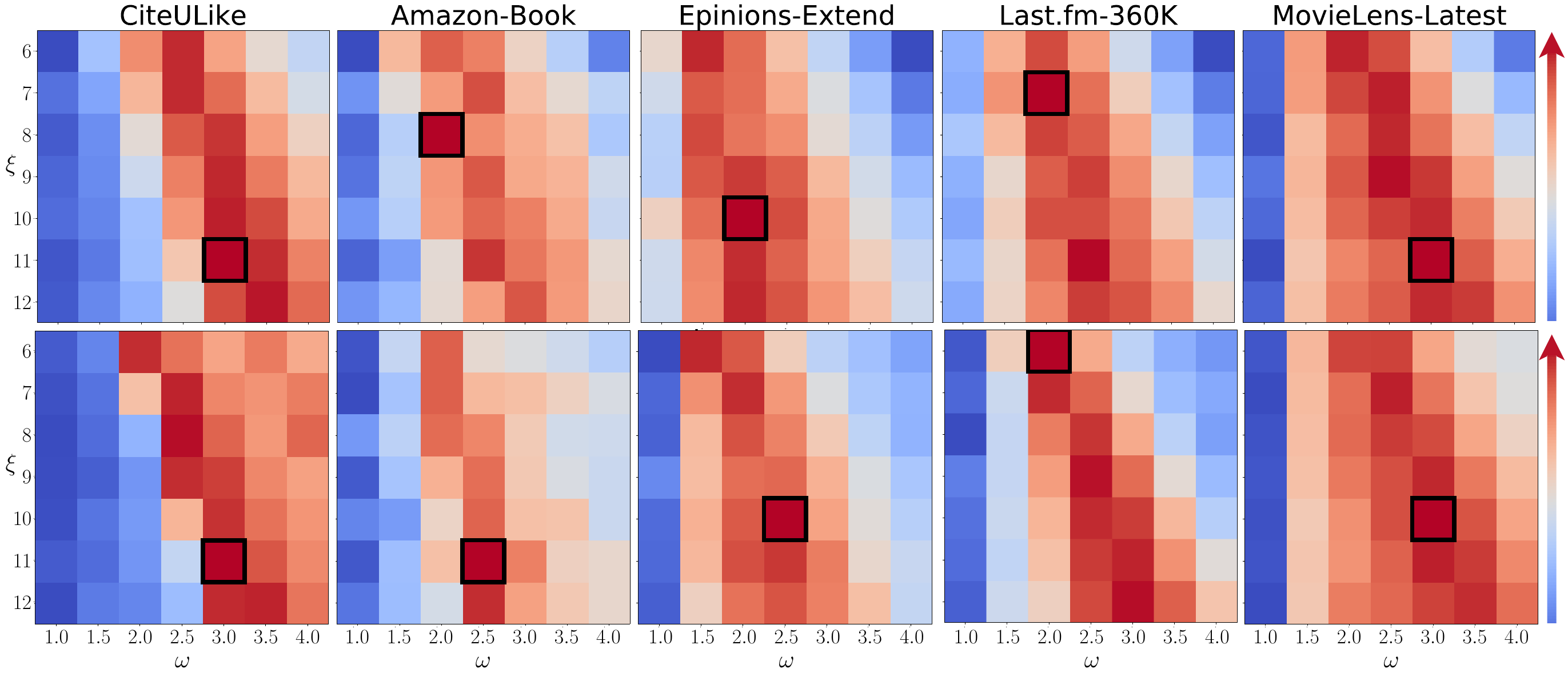}
\vspace{-0.7cm}
\caption{Sensitivity analysis. The first and the second rows represent the results for $\eta=3$ and $\eta=5$, respectively.}
\label{fig:figure1}
\end{center}
\end{figure*}

\begin{figure}
\begin{center}
\includegraphics[width=0.46\textwidth]{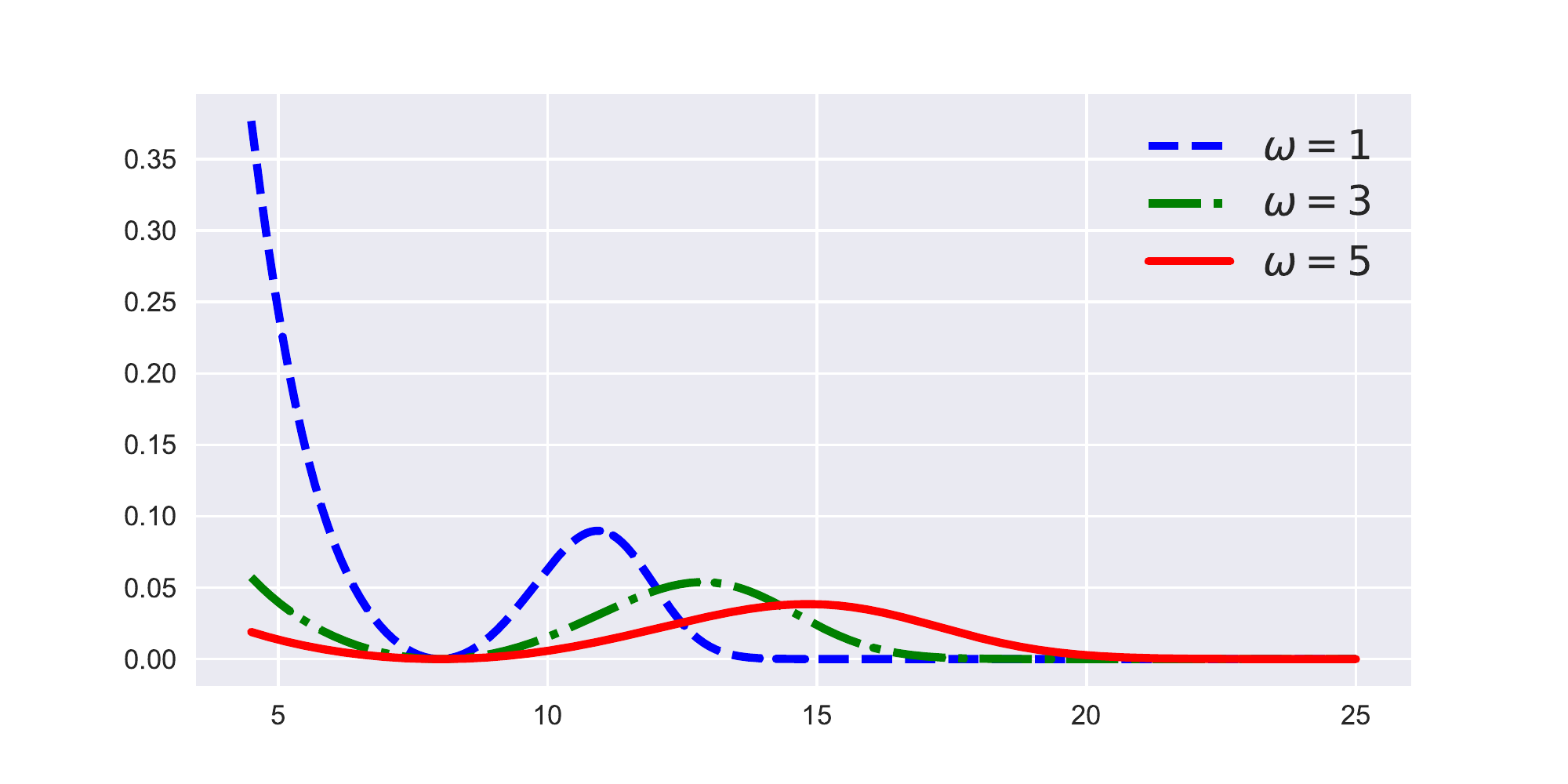}
\vspace{-0.2cm}
\caption{Gradient smoothing ($\xi=8$, $\eta=3$).}
\label{fig:gradient}
\end{center}
\end{figure}

\subsubsection{Top-$N$ Recommendation Performance}
Table~\ref{tab:rec} compares the top-$N$ recommendation performance of Skew-OPT and the five baseline methods, where we list the results with
$\eta=1,3,5$ for comparison; the best results are highlighted in bold. 
For NGCF, we report only the results on Amazon-book and CiteULike due to
computational resource limitations.\footnote{NGCF requires extensive
computational time for large-scale datasets, e.g., more than 24 hours to obtain
a converged result for MovieLens-Latest; note that the training of other models
including ours can however be completed within an hour.}
Note that the~$\dagger$ symbol in the table indicates the best performing
method among all the baseline methods, and the reported percentage improvement
(Improv.\ (\%)) denotes the improvement of the proposed Skew-OPT with respect
to the best-performing baseline. 
Observe from the table that WARP, HOP-rec, and NGCF serve as strong and competitive
baselines.
Even so, the proposed Skew-OPT surpasses all five baselines by a
significant amount for the experiments on all five datasets.  
The results demonstrate that the proposed model maintains
consistent superior performance among different datasets in terms of
both recall@10 and mAP@10, where the improvements range from
3.78\% to 15.25\% in Recall@10 and from 3.18\% to 18.07\% in mAP@10 when $\eta=3$, and from
4.76\% to 14.12\% in Recall@10 and from 3.97\% to 22.53\% in mAP@10 when
$\eta=5$.
Hence, according to the results reported in Table~\ref{tab:rec}, we believe such improvements are substantial, thereby significantly advancing the existing state of the art.
It is also worth mentioning that the proposed Skew-OPT achieves better results than Hop-Rec and
NGCF by solely using user-item interactions without exploring high-order
connections.
\subsubsection{Sensitivity Analysis}
\begin{figure}[h!]
\begin{center}
\begin{subfigure}[b]{0.45\textwidth}
         \centering
         \includegraphics[width=\textwidth]{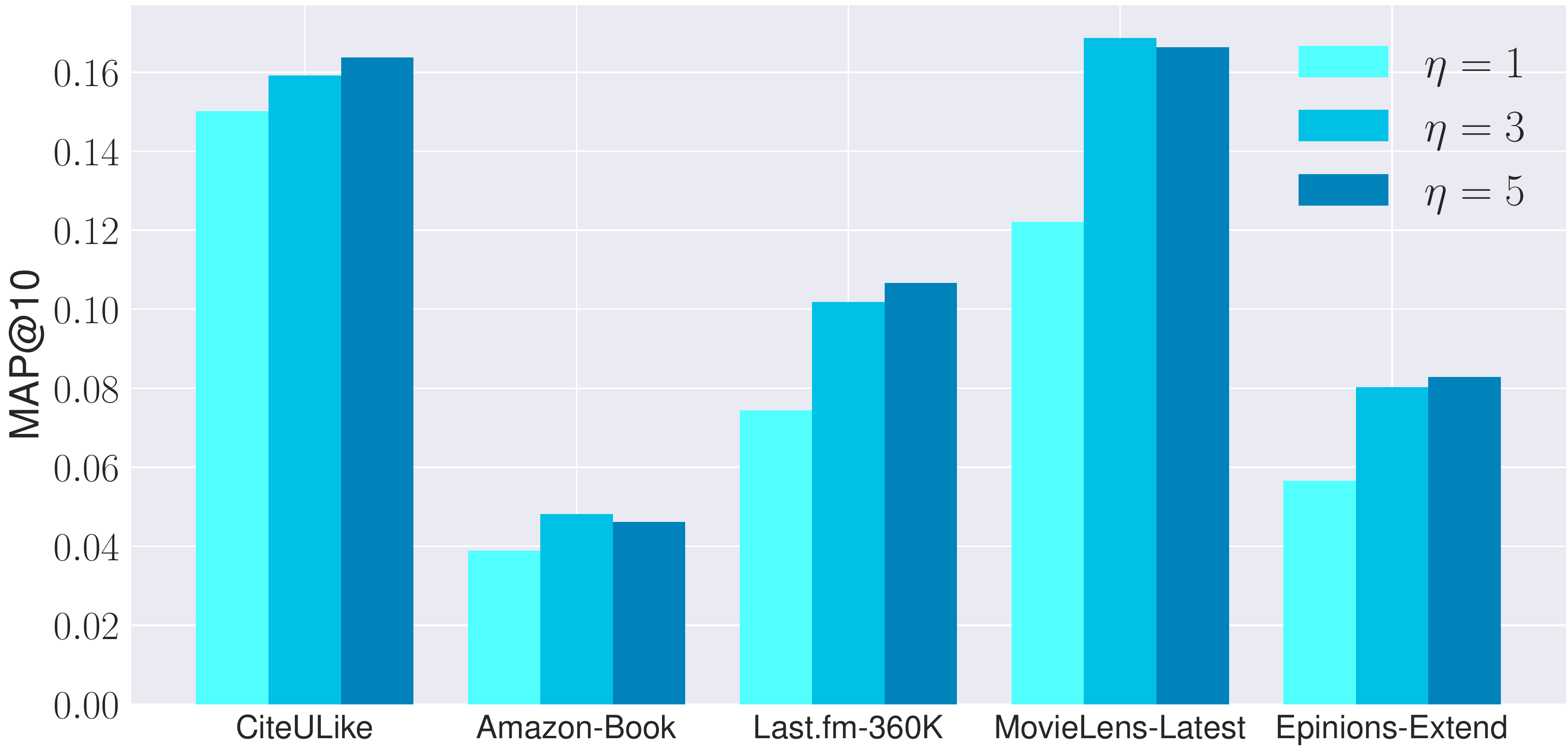}
         \caption{$\xi=11$, $\omega=3$}
\end{subfigure}
\begin{subfigure}[b]{0.45\textwidth}
         \centering
         \includegraphics[width=\textwidth]{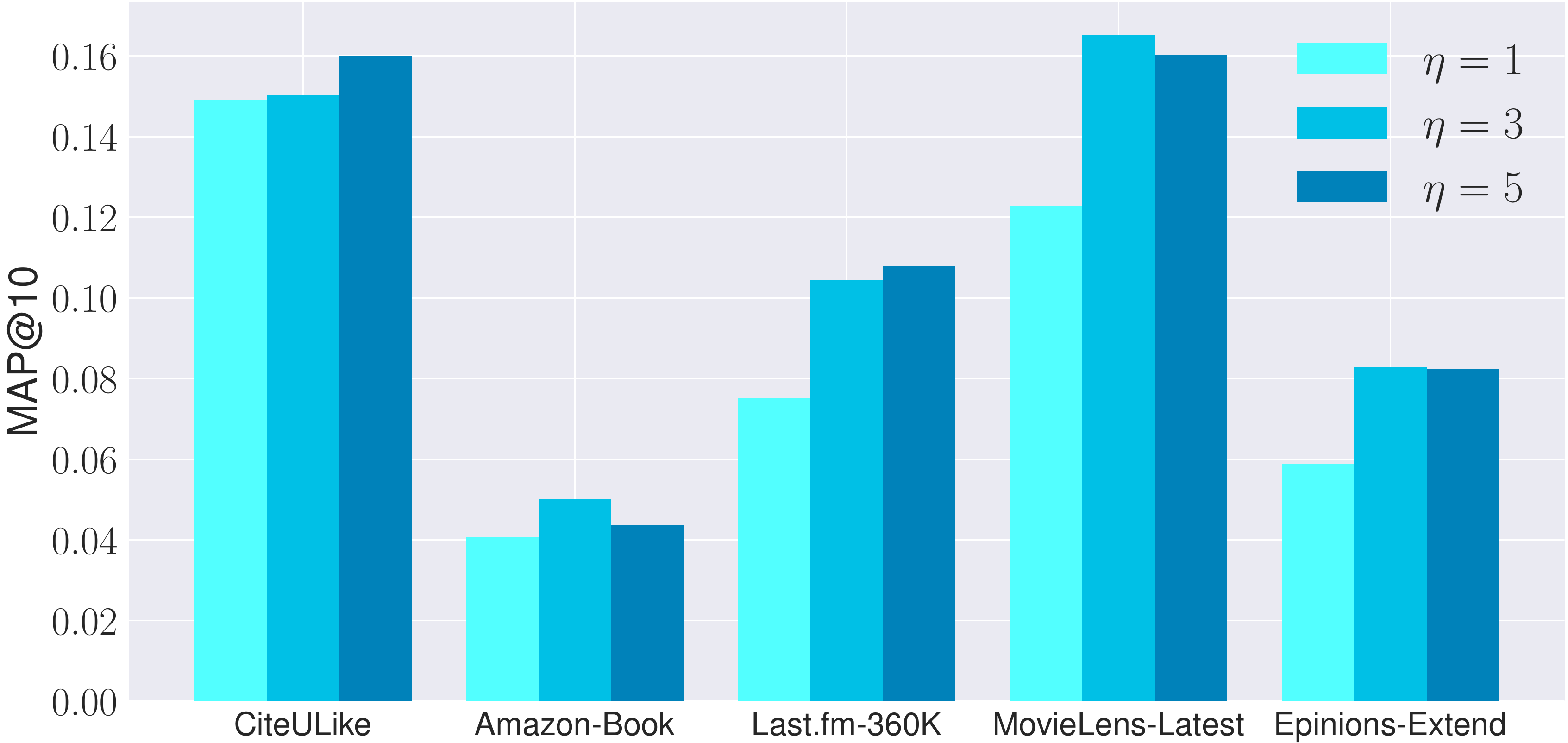}
         \caption{$\xi=12$, $\omega=3$}
\end{subfigure}
\vspace{-0.2cm}
\caption{Sensitivity analysis on $\eta$.}
\label{fig:bar_perform}
\end{center}
\end{figure}

\begin{figure*}
\begin{center}
\includegraphics[width=\textwidth]{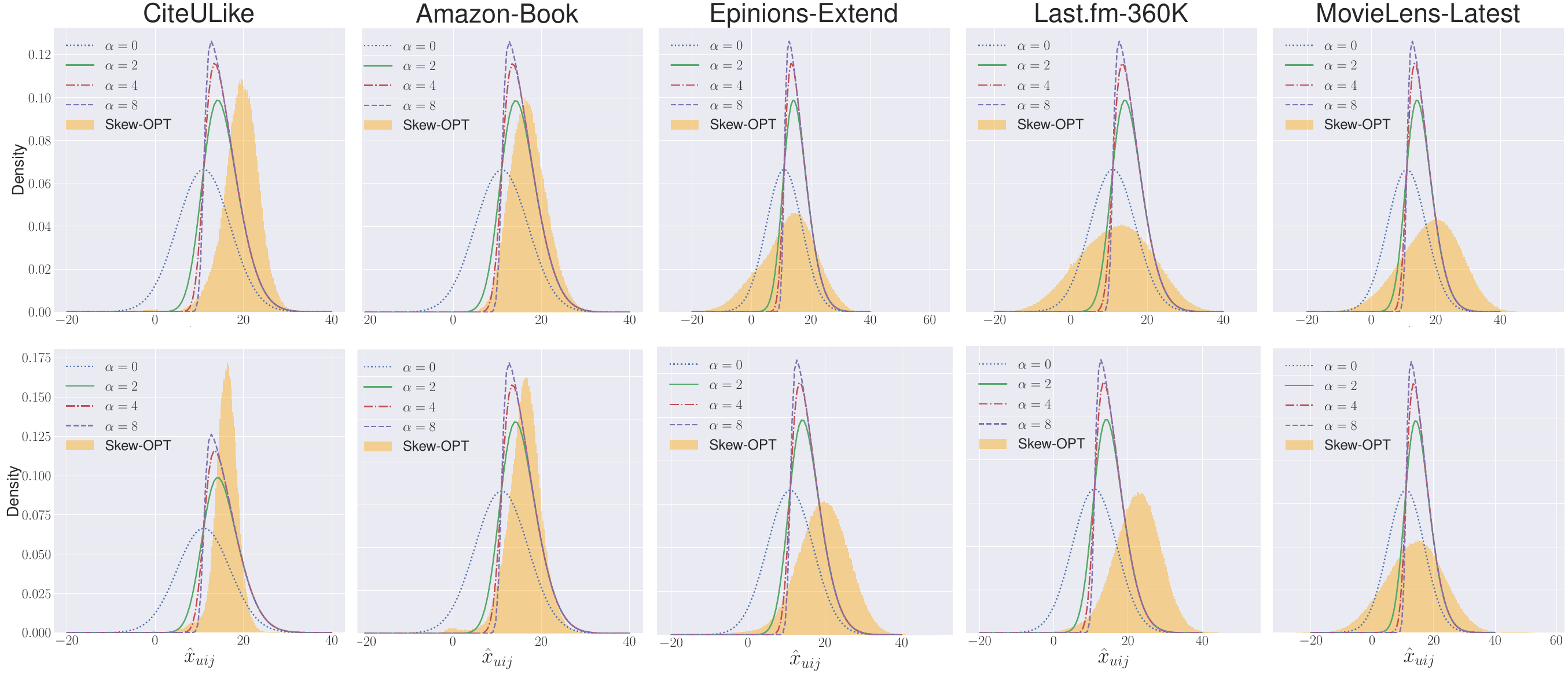}
\vspace{-0.6cm}
\caption{Learned distributions and the skew normal distributions ($\xi=11$, $\omega=3$). The first and the second rows represent the distributions when $\eta=3$ and $\eta=5$, respectively. The bar plots in red denotes the learned distributions with Skew-OPT, and the curves corresponds to the skew normal distributions with $\xi=11$ and $\omega=3$ but with various values of shape parameter $\alpha$.}
\label{fig:distribution}
\end{center}
\end{figure*}
Figure~\ref{fig:figure1} shows the heat maps for mAP@10 on the two key
hyperparameters $\xi$ and $\omega$ in the proposed Skew-OPT; note that we here
plot the results only for $\eta=3,5$ as these two values yield consistently
better performance than $\eta=1$ as shown in Table~\ref{tab:rec}. 
From the figure, we observe that increasing $\xi$, which stands for the
location parameter of the estimator, generally improves the performance while
considering a proper~$\omega$.
In addition, the results of all of the five datasets display a similar
tendency in this sensitivity check; that is, a large $\xi$ usually requires a
large $\omega$ and a small $\xi$ considers a small~$\omega$. 
In other words, if we consider the parameter setting in an opposite direction
from this characteristic, the performance of our model deteriorates.
This is due to the fact that increasing $\xi$ actually increases the possibility
of the model overfitting whereas a large $\omega$ yields gradient smoothing for the
optimization (see Figure~\ref{fig:gradient} which demonstrates the gradient
smoothing effect), thereby better balancing the overfitting that results from a
large $\xi$. 
Note that the square framed in black in each of the sub-figures of
Figure~\ref{fig:figure1} denotes the best performance for each dataset, the
value of which is listed in Table~\ref{tab:rec} (see the values in the columns
for mAP@10 in the table).
We also provide sensitivity checks on $\eta=1,3,5$ with fixed
$\xi=11$ and $\omega=3$ and $\xi=12$ and $\omega=3$ in Figure~\ref{fig:bar_perform}.
The figure shows that under the same location parameter and scale
parameter, $\eta=3$ and $\eta=5$ usually yield better performance than $\eta=1$
among all datasets.


\subsubsection{Distribution Analysis}

Figure~\ref{fig:distribution} compares the learned distribution of the estimator $\hat{x}_{uij}$  for each dataset to the corresponding skew normal distribution under the setting of $\xi=11$ and $\omega=3$.
Note that the learned distribution is generated from the training data with the model trained on the hyperparameters same as the above setting, i.e., $\xi=11, \omega=3$, and $\eta=3$ (the first row) or $\eta=5$ (the second row).
From the figure, we observe that the learned distributions are with similar shapes to the right-skewed normal distributions, especially under the case that $\eta=5$.
It is worth noting that as Skew-OPT does not directly constrain the distribution, there is by nature no guarantee on the shape of the learned distributions.
Moreover, except for the maximization to the likelihood function in the objective function (i.e., the first term in the objective), Skew-OPT also involves a regularization term; as a result, it is nature that the learned distributions do not exactly fit the skew normal distributions with the same $\xi$ and $\omega$. 
Even so, from Figure~\ref{fig:distribution}, we observe that the learned distributions for all datasets are all right-skewed, which corresponds to the statement in Lemma~\ref{lm1} and the AUC analogies in Section~\ref{sec:AUC}.
On the other hand, Figure~\ref{fig:diffxi} shows the learned distributions when adopting different location parameters~$\xi$ but with fixed $\omega=2$ and $\eta=3$. 
As shown in the figure, pushing $\xi$ to be a larger value indeed moves the distribution to the right, thereby increasing the possibility of $\hat{x}_{uij}>0$ and thus the potential to boost the recommendation performance.
\begin{figure}
\begin{center}
\includegraphics[width=0.43\textwidth]{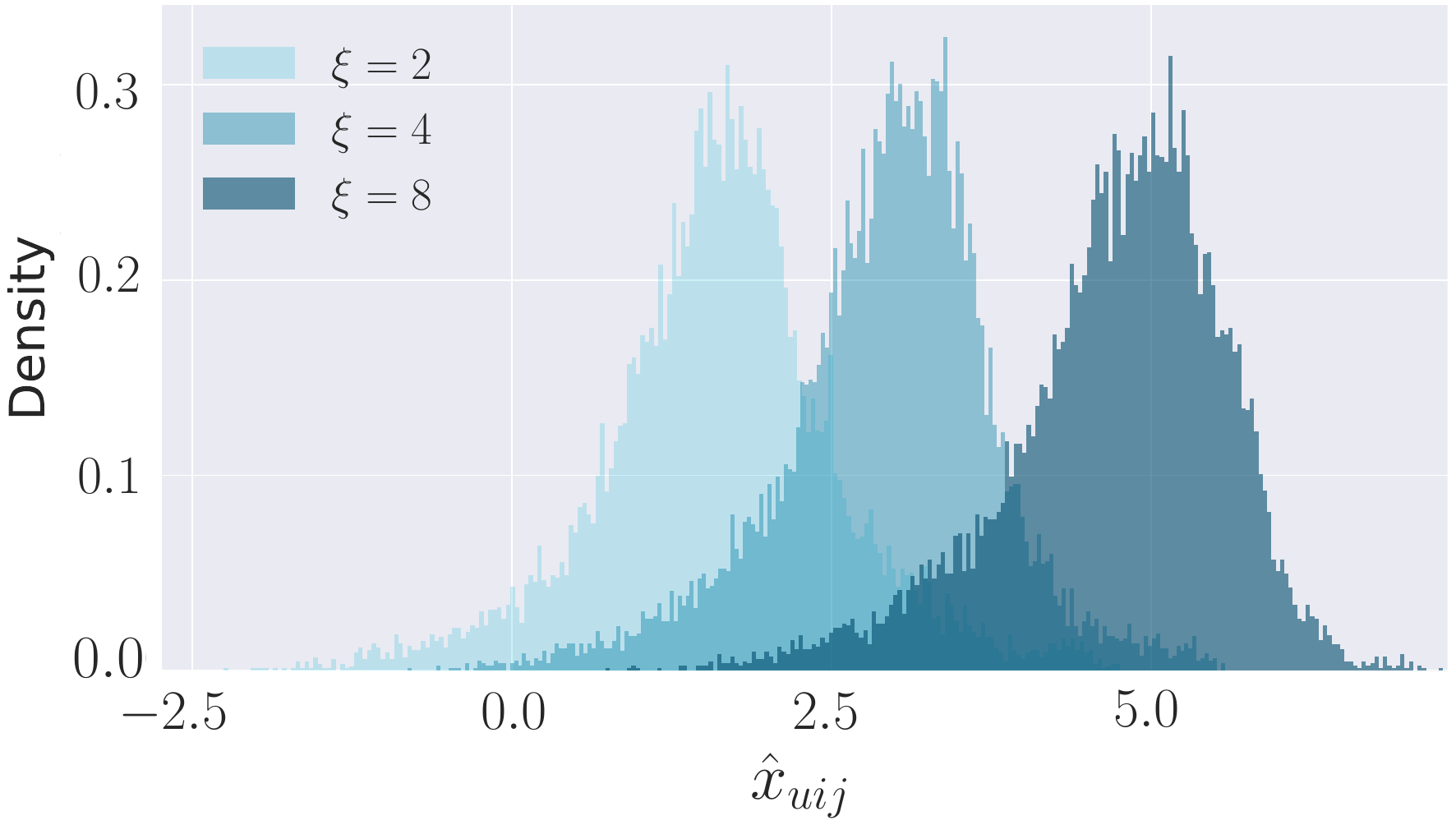}
\vspace{-0.2cm}
\caption{Learned distributions with different location parameters ($\omega=2$ and $\eta=3$).}
\label{fig:diffxi}
\end{center}
\end{figure}

\section{Conclusions}

This paper proposes an unconventional optimization criterion, Skew-Opt, that leverages features of the skew normal distribution to better model the problem of personalized recommendation. 
Specifically, the developed criterion is parameterized with three hyperparameters, thereby providing additional degrees of freedom for ranking optimization.
We further present theoretical insights on the relation between the maximization of Skew-OPT and the shape parameter in the skew normal distribution along with the skewness as well as the asymptotic results of the criterion to AUC maximization.
Experimental results show that models trained with the proposed Skew-OPT yield consistently the best recommendation performance on all tested datasets.
In addition, the sensitivity and distribution analyses not only provide valuable and practical insights for choosing the hyperparameters but also attest the importance of the characteristics of the learned distribution to the recommendation performance.
In sum, this work is the first that explicitly considers the distribution of the estimator for recommendation algorithms; exploring the way to shape the estimator distribution should be of great potential to boost recommendation performance and is an interesting future research direction worth to further investigate.

\bibliographystyle{acm}
\bibliography{paper}

\end{document}